\documentclass[]{llncs}
\pdfoutput=1
\usepackage{latexsym}
\usepackage{epsfig}
\usepackage{color}
\usepackage{wrapfig}

\newtheorem{observation}{Observation}

\newcommand{\indeg}{{\mbox{\it indeg}}}
\newcommand{\outdeg}{{\mbox{\it outdeg}}}

\title{Straightening out planar poly-line drawings
}
\author{Therese Biedl\thanks{Supported by NSERC and the Ross and Muriel
Cheriton Fellowship.}}
\institute{David R.~Cheriton School of Computer Science, 
University of Waterloo,  \\
Waterloo, ON N2L 3G1, Canada, {\tt biedl@uwaterloo.ca}}
\date{\today}

\begin{document}
\maketitle

\begin{abstract}
This paper addresses
the following question:  Given a planar poly-line drawing of a graph, can we
``straighten it out'', i.e., convert it to a planar straight-line drawing,
while keeping some features unchanged?

We show that any $y$-monotone poly-line drawing can be straightened out
while maintaining 
$y$-coordinates and height.  The width may increase much, but we also
show that on some graphs
exponential width is required if we do not want to increase
the height.  
Likewise $y$-monotonicity is required: there are poly-line
drawings (not $y$-monotone)
that cannot be straightened out while maintaining the height.
We give some applications of our result.

\end{abstract}

\section{Introduction}
\label{se:intro}

Let $G=(V,E)$ be a simple graph with $n=|V|$ vertices and
$m=|E|$ edges.  We assume that $G$ is {\em planar},
i.e., it can be drawn without crossings.
Almost by definition $G$ has a planar {\em poly-line drawing},
i.e., a drawing where vertices are points, edges are drawn
as sequences of contiguous line segments and there are no
crossings.  For if any part of a planar drawing of $G$ uses
curved lines, then these can be approximated by polygonal curves.

However, for ease of readability and storing the drawing,
one prefers a {\em straight-line drawing}, i.e., a drawing
where vertices are points and edges are straight-line segments
between their endpoints.
It was known for a long time that any planar graph has a 
straight-line drawing \cite{Fary48,Stein51,Wagner36},
even in an $O(n)\times O(n)$-grid \cite{FPP90,Sch90}.
Many improvements have been developed since, see for example
\cite{DBETT98,NR04}.

This paper addresses the following question:  Given a planar
poly-line drawing $\Gamma$ of a graph $G$, can we ``straighten it out''?
By the above, there exists a straight-line drawing of $G$,
but can we create a straight-line drawing $\Gamma'$ of $G$ that
is similar to $\Gamma$ in some sense?  For example, for some
applications (we list a few in Section~\ref{se:appl}) it is
much easier to create a poly-line drawing; can it be
converted to a straight-line drawing without losing
some key features?

The answer to our main question obviously depends on what is
meant by ``similar''.  One possible way to define this would
be to request that the relative order of vertices stays
the same, i.e., $v$ is to the left / below $w$ in $\Gamma$
if and only if it is in $\Gamma'$.  However, one quickly sees
that then not all poly-line drawings can be straightened out,
see the drawing of $K_4$ in Figure~\ref{fig:K4}(left).

The main result (Theorem~\ref{thm:main}) of this paper 
is that  if a poly-line drawing
is $y$-monotone, then it can be converted into a straight-line
drawing such that $y$-coordinates of vertices
are unchanged, and in particular the height of the drawing
is unchanged.  This result has some applications that we list.
Most importantly,
it allows to derive some height-bounds for drawing styles for
which we are not aware of any direct proof, and it allows to
formulate some NP-hard graph drawing problems as integer programs.

We also give graphs
where two limitations of Theorem~\ref{thm:main} cannot be overcome:  It
is not possible to remove the ``$y$-monotone'' condition without
increasing the height, and it is not possible to 
keep the width to less than exponential without increasing the height.

\medskip\noindent{\bf Historical note:}
At WAOA'12 \cite{Bie-WAOA12}, we claimed 
a result (about straightening out so-called 
flat visibility representations) that would easily imply 
Theorem~\ref{thm:main}.
Unfortunately the algorithm in \cite{Bie-WAOA12}
is incorrect since the resulting drawing may not be planar; 
we review the algorithm and show the error in the appendix.

\section{Preliminaries}


Throughout this paper $G$ denotes a planar graph, $\Gamma$
denotes a planar poly-line drawing of $G$, and $\Gamma'$
denotes the straight-line drawing of $G$ that we try to construct.
Recall that in $\Gamma$ edges are drawn as sequences
of contiguous line segments; a point of transition between the
line segments is called a {\em bend}.  $\Gamma$ is called
{\em $y$-monotone} if for any edge, the $y$-coordinates 
monotonically increases as we walk from one end to the other.
Horizontal segments are allowed.

We call a drawing a {\em grid-drawing} if vertices and bends
have integer coordinates.  A grid drawing is said to have
{\em width $w$} and {\em height $h$} if (possibly after translation)
vertices and bends are placed on the $[1,w]\times [1,h]$-grid.  
The height is thus measured by the number of {\em rows}, 
i.e., horizontal lines with integer $y$-coordinates that are occupied by the 
drawing.  

Our transformation will be such that $\Gamma$ and $\Gamma'$ 
{\em have the same $y$-coordinates}, i.e., any vertex has 
the same $y$-coordinate in both drawings.  
Since we give transformations only
for $y$-monotone drawings, this implies that
an edge $e$ intersects row $r$ in $\Gamma$ if and only if
it intersects row $r$ in $\Gamma'$.

Our transformation actually achieves a stronger property:
$\Gamma$ and $\Gamma'$
have {\em the same left-to-right order of edges and vertices in each row}.  
This means that if we enumerate in row $r$ from left to right the 
segments $s_1,s_2,s_3,\dots$ that are occupied by $\Gamma$ (i.e., each
$s_i$ is either a point containing a vertex, or a bend of an edge,
or a point where an edge segment crosses row $r$, or a line segment 
where an edge is routed along $r$),
and if we similarly enumerate the segments $s_1',s_2',s_3',\dots$ occupied
by $\Gamma'$ in row $r$, then for all $i$ segments $s_i$ and $s_i'$ belong
to the same vertex or the same edge.

\begin{figure}[t]
\hspace*{\fill}
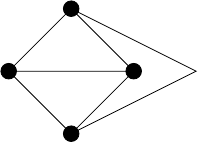
\hspace*{\fill}
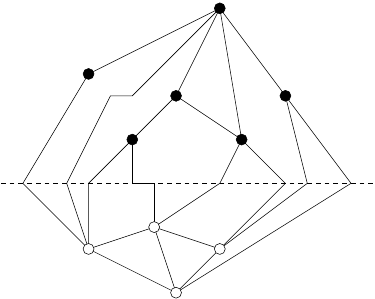
\hspace*{\fill}
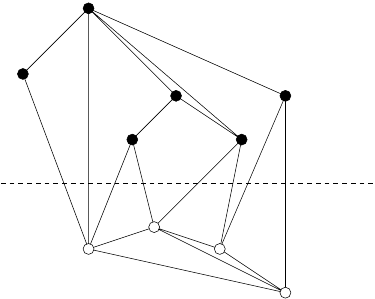
\hspace*{\fill}
\caption{$y$-monotone poly-line drawings. (Left) $K_4$ cannot
be straightened out while maintaining relative $x$-
and $y$-coordinates.  (Middle) An HH-drawing (see 
Section~\ref{se:appl}).  (Right) The HH-drawing straightened out
while maintaining
$y$-coordinates and left-to-right-orders in each row.}
\label{fig:K4}
\label{fig:HHdrawing}
\label{fi:HHdrawing}
\end{figure}

\begin{lemma}
\label{lem:planar}
Let $\Gamma, \Gamma'$ be $y$-monotone poly-line drawings of the same
graph $G$.   If $\Gamma$ is planar, and $\Gamma'$ has the same $y$-coordinates 
and the same left-to-right-orders in each row as $\Gamma$, then $\Gamma'$
is also planar, with the same planar embedding and outer-face.
\end{lemma}
\begin{proof}(Sketch) If edges $e$ and $e'$ crossed in $\Gamma'$,
but not in $\Gamma$, then in the row above or below the point of
the crossing the left-to-right orders are not the same in 
$\Gamma$ and $\Gamma'$.\qed
\end{proof}

\section{Creating straight-line drawings}
\label{se:PL_SL}

In this section, we show how to convert any planar $y$-monotone
poly-line drawing into a planar straight-line drawing.

\subsection{Straightening out triangulated drawings}

We first show such a transformation for a {\em triangulated graph}, i.e.,
for a graph where all faces including the outer-face are triangles.  
It will be helpful to direct an edge $(v,w)$ as $v\rightarrow w$ if $y(v)<v(w)$,
and replace horizontal edges by two directed edges in opposite directions.
Recall that an {\em inner vertex} is a vertex that is not on the outer-face.


\begin{observation}
Let $\Gamma$ be a planar $y$-monotone poly-line drawing of a triangulated
graph $G$.  Then
$\indeg(v)\geq 1$ and $\outdeg(v)\geq 1$ for every inner vertex.
\end{observation}
\begin{proof}(Sketch) If $\indeg(v)=0$ for an inner vertex $v$,
then the face ``below'' $v$ would have degree $\geq 4$ since edges
are drawn $y$-monotonically. \qed
\end{proof}

The conversion into a straight-line drawing will happen by repeatedly
contracting one vertex that has properties that allow to re-insert it later.
The existence of such a vertex is proved in the following lemma.

\begin{lemma}
\label{lem:cases}
Let $\Gamma$ be a planar $y$-monotone poly-line drawing of a triangulated
graph $G$.  Then there exists an inner  vertex $v$ such that
\begin{itemize}
\item $v$ is incident to a horizontal edge, or
\item $\indeg(v)=1$, and $\indeg(w)\geq 2$ for all vertices $w$ with 
	$v\rightarrow w$, or
\item $\outdeg(v)=1$, and $\outdeg(w)\geq 2$ for all vertices $w$ with 
	$w\rightarrow v$.
\end{itemize}
\end{lemma}
\begin{proof}
Since the outer-face contains three vertices, one of the bottom and the
top row (say the bottom row) 
contains only one vertex.  Let $v_1$ be the inner  vertex
that minimizes the $y$-coordinate among inner  vertices, breaking ties
arbitrarily.  

We know $\indeg(v_1)\geq 1$ since $v_1$ is not on the outer-face.
At most one vertex is strictly below $v_1$ by choice of $v_1$ and
since the bottom row contains only one vertex.
If $\indeg(v_1)\geq 2$, then $v_1$ hence has an incident horizontal edge 
and we are done.  Likewise we are done if $\indeg(v_1)=1$  and 
$\indeg(w)\geq 2$ for all $v_1\rightarrow w$.  So assume that $\indeg(w)=1$
for some $v_1\rightarrow w$, and set $v_2:=w$.  Observe that $v_2$ is an
inner  vertex, since any outer-face vertex $x$ with $y(x)\geq y(v)$ 
has another incoming edge from the vertex on the bottom row.
Repeat with $v_2$: either
it is incident to a horizontal edge, or all its successors have indegree 
$\geq 2$, or it has a successor with indegree 1, which we set to be $v_3$.
Repeat.  Eventually this must find a suitable vertex, since otherwise there
would be an infinite sequence $v_1,v_2,v_3,\dots$
with $y(v_1)<y(v_2)<\dots$.
\qed
\end{proof}

The proof of the following lemma gives the main part of the algorithm 
to straighten out $y$-monotone poly-line drawings:

\begin{lemma}
\label{lem:main}
Let $\Gamma$ be a planar $y$-monotone poly-line drawing of a
triangulated graph $G$.  Then there exists a planar straight-line
drawing $\Gamma'$ of $G$ that has the same $y$-coordinates and
the same left-to-right orders as $\Gamma$.
\end{lemma}
\begin{proof}  
If $n=3$, then $G$ is a triangle and the result is easily shown.
Now assume $n\geq 4$.  We have three cases:

\smallskip\noindent{\bf Case 1: } $G$ contains a {\em separating triangle},
i.e., a triangle $\{u,v,w\}$ such that there are other vertices both 
inside and outside $\{u,v,w\}$.

Let $G_0$ be the graph induced by $\{u,v,w\}$ and all vertices outside it,
and let $G_1$ be the graph induced by $\{u,v,w\}$ and all vertices inside it.
Recursively apply the lemma to these two graphs, with the drawing the one
induced by $\Gamma$, to obtain the straight-line drawings $\Gamma_0$ of $G_0$
and $\Gamma_1$ of $G_1$.  The transformation preserves $y$-coordinates 
and left-to-right-orders for the two copies of $\{u,v,w\}$. 
After a horizontal translation and a horizontal
shear, the drawing $\Gamma_1$ can hence be inserted into the face $\{u,v,w\}$
in the drawing $\Gamma_0$ to give a straight-line drawing of $G$.  All
properties are easily verified.

\smallskip\noindent{\bf Case 2: }
$\Gamma$ contains an inner  vertex $v$ with an incident horizontal edge 
$(v,w)$.  Let $w,z_1,x_1,\dots,x_d,z_2$ be the neighbours of $v$ in
clockwise order; note that $z_1,z_2$ are also neighbours of $w$ since
the faces incident to $(v,w)$ are triangles.  If any $x_i$ is a neighbour
of $w$ then $\{v,w,x_i\}$ is a separating triangle and we are
done by Case~1.
So assume that $v,w$ only have $z_1,z_2$ as common neighbours.
Contract $v$ into $w$ and delete the duplicate copies 
of $(w,z_1)$ and $(w,z_2)$; this then yields a simple graph $G_0$.
Create a
poly-line drawing $\Gamma_0$ of $G_0$ 
by routing the edge from $x_i$ to $w$ along the route of $(x_i,v)$
until the row above/below $v$, and then to $w$.
See Figure~\ref{fig:case2}(left).

By induction transform $\Gamma_0$ into a planar straight-line drawing $\Gamma_0'$
with the same $y$-coordinates and same left-to-right-orders.    
Deleting the edges $(w,x_i)$ 
from $\Gamma_0'$ leaves a polygon $P$ 
defined by $w,z_1,x_1,\dots,x_d,z_2,w$ into which
we must insert $v$.  The boundary of $P$ consists of edges of $G_0$
since $G_0$ is triangulated.  Vertex $w$ is adjacent to all vertices of $P$,
hence $P$ is star-shaped with $w$ in its kernel.  Since $(w,x_i)$
(for $i=1,\dots,d$) was drawn in $\Gamma_0'$ without overlapping,
it belongs to the interior of $P$.  Therefore the kernel actually
contains an open region around $w$.  Inside this open region we can
find a point $p(v)$, with $y$-coordinate $y(v)$ and to the right of $w$,
where we place $v$.  Since $v$ is in the kernel of $P$, we can then connect
all its edges without introducing crossings and obtain $\Gamma'$.  By
construction the $y$-coordinate of $v$ is the same as in $\Gamma$,
and one easily verifies that left-to-right-orders are the same in all
rows the clockwise order of edges at $v$ is correct in $\Gamma'$.

\begin{figure}[t]
\hspace*{\fill}
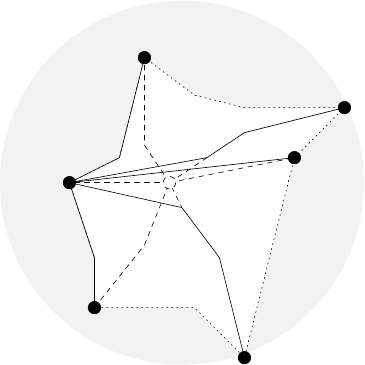
\hspace*{\fill}
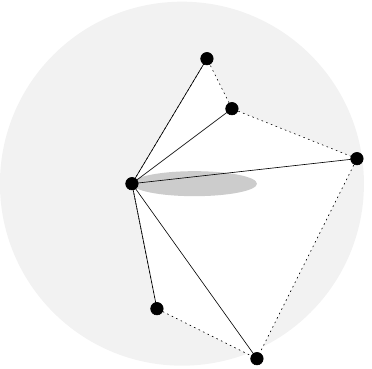
\hspace*{\fill}
\caption{Case 2: An inner vertex $v$ has a horizontal edge $(v,w)$.
(Left) Contracting $v$ into $w$ and how to re-route edges.  (Right)
We can insert $v$ in the kernel.}
\label{fig:case2}
\end{figure}

\smallskip\noindent{\bf Case 3: }  None of the above.
By Lemma~\ref{lem:cases}, there exists an inner  vertex $v$ with
$\indeg(v)=1$ and $\indeg(w)\geq 2$ for all $v\rightarrow w$. (The
case of a vertex $v$ with $\outdeg(v)=1$ and $\outdeg(w)\geq 2$ for
all $w\rightarrow v$ is symmetric.)

Let $u,x_1,\dots,x_d$ be the neighbours of $v$ in clockwise order, 
where $u \rightarrow v$ is the unique incoming edge of $v$.  
Since Case 2 does not apply, $v$ has no incident horizontal edge
and so $y(u)<y(v)<\min\{y(x_1),y(x_d)\}$.
Assume for contradiction that
$y(x_{i-1})>y(x_i)<y(x_{i+1})$ for some $1< i<d$, i.e., that there
exists a local minimum (with respect to $y$-coordinates) among the
neighbours above $v$.  Then $\indeg(x_i)=1$, since the incoming edges
are consecutive in the cyclic order, but the edges before and after
$v\rightarrow x_i$ at $x_i$ are both outgoing.  This contradicts the
choice of $v$.  It follows that 
%
$$y(v)<y(x_1)\leq y(x_2)\leq \dots\leq y(x_m)\geq y(x_{m+1})\geq \dots 
\geq  y(x_d)>y(v)$$
%
for some $1\leq m\leq d$.
See Figure~\ref{fig:Case3}(left).
Let $G_0$ be the graph obtained from $G$ by contracting $v$ into $u$
and deleting the duplicate copies of $(u,x_1)$ and $(u,x_d)$ that
result.
%
This does not create multiple edges since the only common neighbours of
$u$ and $v$ are $x_1$ and $x_d$ (otherwise we'd have a separating triangle
and be in Case 1.)
Let $\Gamma_0$ be the poly-line drawing obtained from $\Gamma$ by deleting
$v$ and then inserting the above edges; this can be done by drawing 
$(u,x_i)$ near the drawings of $(u,v)$ and $(v,x_i)$ for $i=1,\dots,d$.  

\begin{figure}[t]
\hspace*{\fill}
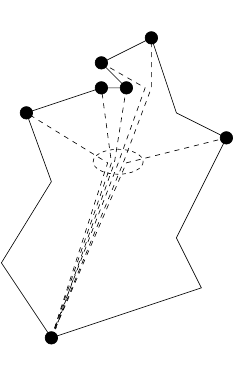
\hspace*{\fill}
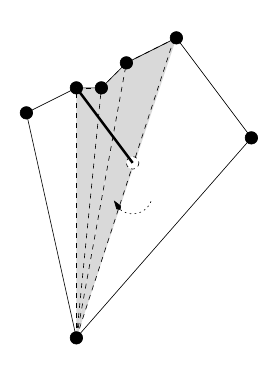
\hspace*{\fill}
\caption{Case 3: $\indeg(v)=1$ and $\indeg(w)\geq 2$ for $v\rightarrow w$.
(Left) Contracting $v$ into $u$.
(Right) Point $p(v)$ can see $x_i$.
}
\label{fig:Case3}
\end{figure}

By induction we can convert $\Gamma_0$ into a planar straight-line drawing
$\Gamma_0'$ with the same $y$-coordinates and left-to-right-orders.  Delete
the added edges from $\Gamma_0'$, and let $p(v)$ be the point where the
edge $(u,x_m)$ intersected the row with $y$-coordinate $y(v)$.    
We claim that point $p(v)$ can see all of $u,x_1,\dots,x_d$:
\begin{itemize}
\item Clearly $p(v)$ can see $u$ and $x_m$ since it is placed on where edge
$(u,x_m)$ was drawn as a straight line.
\item For any $1 \leq i < m$ consider the polygon $F_i$ defined by
the union of the faces $\{u,x_i,x_{i+1}\},\dots,\{u,x_{m-1},x_m\}$ 
in $\Gamma_0'$.  Since these were faces, and we deleted the edges
$(u,x_{i+1}),\dots,(u,x_m)$, polygon $F_i$ has an empty interior.  
Since $y(p(v))=y(v)<y(x_i)$, while $y(x_i)\leq y(x_{i+1} \leq \dots
\leq y(x_m)$, the line segment from $p(v)$ to $x_i$ is 
inside $F_i$,
so $p(v)$ can see $x_i$.
\item Similarly we prove that $p(v)$ can see any $x_i$ with $m<i\leq d$.
\end{itemize}
Hence placing $v$ at $p(v)$ gives a planar drawing that respects
the $y$-coordinate and planar embedding at $v$.
\qed
\end{proof}

One interesting consequence of this proof is that for any
$y$-monotone planar poly-line drawing of height $h$, there exists
a pair of vertices such that contracting them yields again
a planar $y$-monotone poly-line drawing of height $h$.  This contrasts with
a result in \cite{FLW03} that having a planar drawing of height $h$
is {\em not} always preserved under contracting (arbitrary) edges.

\subsection{Triangulating $y$-monotone poly-line drawings}

For space
reasons we only sketch here how to triangulate $y$-monotone poly-line
drawings; the appendix gives all details.
As a first step, add one new row each at the top and the bottom.
	Add one new vertex in the new bottom row and two new vertices in
	the new top row, and connect them, using $y$-monotone curves,
	 with a triangle that encloses
	the entire drawing.  

As next step, ensure that every inner vertex $v$ has a neighbour $w$
	with $y(w)>y(v)$.  If $v$ has none, then go upward from $v$
	until we hit an edge $e$ (this must happen since $v$ is not on
	the outer-face).  Let $w$ be the head of $e$, and note that
	we can route the new edge $(v,w)$ by going upward from $v$
	and then going parallel to $e$.
	Similarly ensure that each inner vertex $v$ has a neighbour with
	a smaller $y$-coordinate.

Now every inner face $f$ is drawn $y$-monotone, with the
	only horizontal edges on $f$ at its minimum or maximum $y$-coordinate.
	One can argue that under this restriction {\em any} 
	two non-adjacent vertices of $f$ can be connected with a 
	$y$-monotone path inside $f$.  
	Hence we can triangulate $f$ in any way that does not insert multiple
	edges and eventually
	obtain a $y$-monotone poly-line drawing of a triangulated graph.

Convert the drawing into a straight-line drawing as in Lemma~\ref{lem:main},
then delete the added edges and vertices.  The added rows are then 
	empty (since $y$-coordinates are preserved) and so can also be
	deleted.  So we have:

\begin{theorem}
\label{thm:PL_SL}
\label{thm:main}
\label{th:main}
Any planar $y$-monotone poly-line drawing $\Gamma$ can be transformed into a 
planar straight-line drawing $\Gamma'$ with the same $y$-coordinates
and the same left-to-right orders in each row.
\end{theorem}

The proof of this theorem is algorithmic, and clearly leads to a 
quadratic-time algorithm.  Reducing this run-time remains an
open problem.  It is not hard to build a data structure to
find a vertex $v$ as in Lemma~\ref{lem:cases} in amortized constant
time, but it is not clear how we can test in less than linear time
per vertex $v$ 
whether $v$ belongs to a separating triangle, because the graph
changes due to the edge contractions.

\section{Optimal height means exponential width}

While our transformation in Theorem~\ref{thm:PL_SL}
keeps the height intact,
the width can increase dramatically.    Our construction does not
give a grid drawing, but it is clear that with
minor modifications one can achieve rational coordinates,
hence integer coordinates after horizontal scaling.
We do not know upper bounds on these grid coordinates, but we can
argue that they are at least exponential for some graphs.
To do so, we first study special drawings of one graph:


\begin{figure}[t]
\hspace*{\fill}
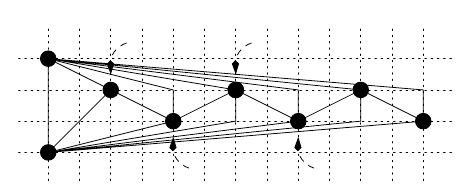
\hspace*{\fill}
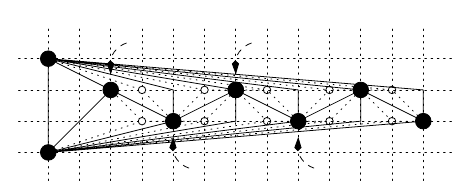
\hspace*{\fill}
\caption{(Left) A planar graph. (Right) Inserting vertices into inner faces.}
\label{fig:bad_example}
\label{fig:bad_bad}
\end{figure}

\begin{lemma}
\label{lem:bad}
Let $G$ be the graph shown in 
Figure~\ref{fig:bad_example}(left). 
Then any planar straight-line drawing  $\Gamma$ that respects
the $y$-coordinates and left-to-right-orders of 
Figure~\ref{fig:bad_example}
has width at least $\frac{1}{3}2^{n-1}$.
\end{lemma}
\begin{proof}
Denote by $x(w)$ the $x$-coordinate of vertex $w$ in drawing $\Gamma$.
We will assume that $x(v)\leq x(u)$; the other case is proved similar
and in fact gives an even larger  width bound.  After possible translation,
we may also assume $x(v)=0$.  
We will show by induction on $i$ that 
$$
x(a_{2i-1})\geq \frac{1}{3}(x(u)+2^{2i})-1 
\quad\mbox{ and }\quad
x(a_{2i})\geq \frac{1}{3}(2x(u)+2^{2i+1})-1
$$
for $i\geq 1$; this implies the result since the width is then
$x(a_d)-x(v)+1\geq \frac{1}{3}(x(u)+2^{d+1})\geq \frac{1}{3}2^{d+1} = \frac{1}{3}2^{n-1}$. 

\begin{figure}[t]
\hspace*{\fill}
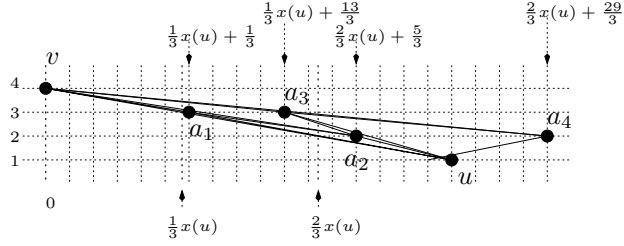
\hspace*{\fill}
\caption{Computing the required width 
in a straight-line drawing.}
\label{fig:bad_example_2}
\end{figure}

Consider vertex $a_1$, which is placed on row $3$.  
The straight-line segment $(u,v)$ crosses row $3$ at $x$-coordinate
$\frac{1}{3}x(u)$, and $a_1$ must be to the right of that, so
$$x(a_1)\geq \lfloor \frac{1}{3}x(u) \rfloor + 1 \geq \frac{1}{3}x(u)
+\frac{1}{3} = \frac{1}{3}(x(u)+2^2)-1$$
as desired.  Now consider vertex $a_{2i}$ for $i\geq 1$.  The line segment
of edge $(a_{2i},v)$ crosses row $3$ at $x$-coordinate
$\frac{1}{2}x(a_{2i})$.
Since left-to-right-orders are preserved, this
crossing must be to the right of $a_{2i-1}$, therefore
$x(a_{2i-1})<\frac{1}{2}x(a_{2i})$ or $2x(a_{2i-1}) < x(a_{2i})$.
By integrality therefore
$$x(a_{2i}) \geq 2x(a_{2i-1})+1 \geq \frac{1}{3}(2x(u)+2\cdot 2^{2i}) - 2 + 1$$
as desired.
Finally consider vertex $a_{2i+1}$ for $i\geq 3$.  The line
segment of edge $(a_{2i+1},u)$ crosses row 2 at $x$-coordinate
$(x(u)+x(a_{2i+1}))/2$.  Since left-to-right-orders
are preserved, this crossing must be to the right of $a_{2i}$, therefore
$x(a_{2i})<\left(x(u)+x(a_{2i+1})\right)/2$ or
$2x(a_{2i})-x(u)<x(a_{2i+1})$.  By integrality therefore
$$x(a_{2i+1})\geq 2x(a_{2i}) -x(u) + 1 \geq 
\frac{1}{3}(4x(u)+2\cdot 2^{2i+1})-2-x(u)+1
=\frac{1}{3}(x(u)+ 2^{2i+2})-1$$
as desired.
%
\qed
\end{proof}

\begin{theorem}
\label{thm:bad}
There exists a graph $H$ that has a planar straight-line drawing on four rows,
but any planar straight-line drawing on four rows has width 
at least $\frac{1}{3}2^{n/3}$.
\end{theorem}
\begin{proof}
The graph $H$ is obtained by taking the graph $G$ 
from Figure~\ref{fig:bad_example}(left)
with $d\geq 11$ and inserting into each inner face except $\{u,v,a_1\}$
a new vertex 
adjacent to the three vertices of the face.  
Note that $H$ is triangulated and has $3d$ vertices. 
It has a $y$-monotone poly-line drawing on four rows
(see Figure~\ref{fig:bad_example}(right)), and hence by Theorem~\ref{thm:main}
also a straight-line drawing on 4 rows.

Let $\Gamma_H$ be an arbitrary
planar straight-line drawing of $H$ that uses four rows. 
Let $\Gamma_G$ be the induced planar straight-line drawing of $G$.
The goal is to show that $\Gamma_G-a_d$ satisfies the conditions of
Lemma~\ref{lem:bad}.

{\smallskip\noindent \bf Claim:} 
{\it Edge $(u,v)$ connects the top and bottom row. }
For observe that triangle $\{u,v,a_{3i}\}$ separates 
triangle $\{a_{3i-2},a_{3i-1},x\}$ from triangle
$\{a_{3i+1},a_{3i+2},x'\}$, where $i=1,2,3$ and $x,x'$ are 
suitable vertices from $H-G$.   Regardless of the choice
of outer-face hence triangle $\{u,v,a_{3i}\}$ surrounds
another triangle and must contain vertices in rows 1 and 4
by $y$-monotonicity.  If (say) row 1 contains neither $u$ nor $v$,
then it must hence contain $a_3$, $a_6$ and $a_9$, which means
that we can construct a planar drawing of $K_{3,3}$ by adding
another vertex in row 0 and connecting it to $a_3,a_6,a_9$.
This is impossible, so one of $\{u,v\}$ is in row 1 and the
other in row 4.

\smallskip
If there were vertices both left and right of edge $(u,v)$, then $\{u,v\}$
would be a cutting pair,  which contradicts that $H$ is triangulated. 
So there are no vertices to one side of $(u,v)$. 
After possible horizontal flip of the drawing, 
renaming of $\{u,v\}\rightarrow \{v,u\}$,
and 
renaming $\{a_1,\dots,a_d\}\rightarrow \{a_d,\dots,a_1\}$,
we may assume that $u$ is on row $1$, $v$ is on row 4, the
remaining vertices are to the right of edge $(u,v)$, and the
outer-face is $\{u,v,a_d\}$.  i.e., the same as in 
Figure~\ref{fig:bad_example}

{\smallskip\noindent \bf Claim:} 
{\it For any $1\leq i<d$, vertex $a_i$ is not
on row 1 or 4.}  For if it were on row 1, 
then edge $(u,a_i)$ would be horizontal.  
Edge $(u,a_{i+1})$ comes after $(u,a_i)$
in clockwise order around $u$ in $G$, but is (by the above) to the
right of $u$, which is impossible since there is no lower row.
Similarly $a_i$ is not on row 4 due to edges $(v,a_i)$ and
$(v,a_{i+1})$.

{\smallskip\noindent \bf Claim:} 
{\it For any $1\leq i<d$, vertices $a_i$ and
$a_{i+1}$ are on different rows.}  For assume $a_i$ and $a_{i+1}$ are both
on row 2 (the case of row 3 is symmetric).
Then triangle $\{u,a_i,a_{i+1}\}$
is drawn on two adjacent rows and hence has no grid-point in its interior,
contradicting that in $H$ there exists a vertex inside $\{u,a_i,a_{i+1}\}$.

\smallskip
If necessary, flip $\Gamma_H$ upside down and 
rename $\{u,v\}\rightarrow \{v,u\}$
so that $a_1$ is on row 3.  Therefore $a_2$ is on row 2,
$a_3$ is on row 3, 
and generally $a_i$ is on row 2 for $i<d$ even and on row 3 for $i<d$ odd.
Hence $\Gamma_G-a_d$ is drawn with exactly the $y$-coordinates as in 
Figure~\ref{fig:bad_example}. It also has the same left-to-right-orders,
since the planar embedding is unique.
By Lemma~\ref{lem:bad}, therefore
$\Gamma_G-a_d$ has width at least $\frac{1}{3}2^d=\frac{1}{3}2^{n/3}$.
$\Gamma_H$ can have no smaller width, which proves the theorem.
\qed
\end{proof}

By the first of these claims, the graph requires 4 rows in {\em any}
straight-line drawing.  As a consequence of Theorem~\ref{thm:bad}, we
hence have:

\begin{corollary}
There exists a planar graph such that any planar straight-line drawing 
that has optimal height has exponential area.
\end{corollary}

\section{Optimal height for poly-line vs.~straight-line}

Theorem~\ref{thm:main} required $y$-monotonicity of the poly-line drawing.
One can show that this condition cannot be dropped unless we allow
an increase in height and a change of $y$-coordinates.


\begin{theorem}
\label{thm:PL_monPL}
There exists a graph with a planar poly-line drawing on 6 rows that
has no planar $y$-monotone poly-line drawing on 6 rows.
\end{theorem}

The graph and its poly-line drawing on 6 rows are shown in 
Figure~\ref{fig:not_monotone_nolabel}, and details of the proof are
in the appendix.  Since straight-line drawings are $y$-monotone
poly-line drawings,
we hence have that poly-line drawings are sometimes better than
straight-line drawings:

\begin{corollary}
There exists a graph with a planar poly-line drawing on 6 rows that
has no straight-line drawing on 6 rows.
\end{corollary}

\begin{figure}[t]
\hspace*{\fill}
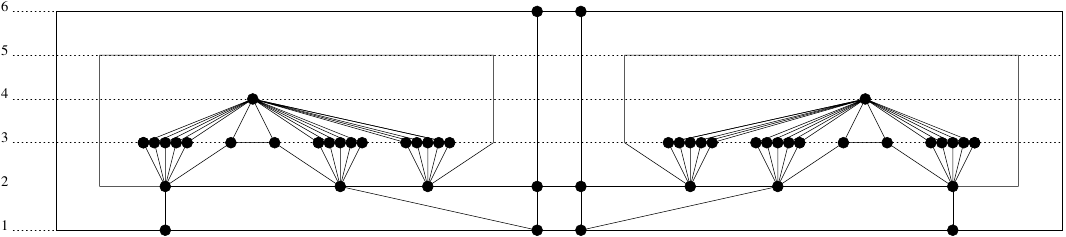
\hspace*{\fill}
\caption{A graph that can be drawn on 6 rows, but not if edges must
be $y$-monotone.}
\label{fig:not_monotone_nolabel}
\end{figure}

\section{Applications}
\label{se:appl}

We give a few applications of Theorem~\ref{thm:main}.

\smallskip\noindent{\bf HH-drawings: }
In a previous paper \cite{BKM-WG98} we studied {\em HH-drawings},
where we are given a planar graph $G$ with a vertex partition $V=A\cup B$,
and we would like to draw $G$ such that all vertices in $A$ have positive
$y$-coordinates and all vertices in $B$ have negative $y$-coordinates.
See also Figure~\ref{fig:HHdrawing}.  We gave necessary and sufficient
conditions for the existence of HH-drawings that were $y$-monotone poly-line
drawings.  We also argued that straight-line HH-drawings required exponential
area for some graphs, but we were not able to actually construct
straight-line HH-drawings.  With Theorem~\ref{th:main}, this missing
link has now been added, because the $y$-monotone poly-line HH-drawing
can be converted to a straight-line drawing, and it is still an HH-drawing
since $y$-coordinates have not been changed.  In particular we hence have:

\begin{theorem}
Any planar bipartite graph has a planar straight-line HH-drawing.
\end{theorem}

\smallskip\noindent{\bf Drawing outer-planar graphs with small height: }
An {\em outer-planar graph} is a planar graph that can be drawn such
that all vertices are on the outer-face.
A {\em flat visibility
representation} is an assignment of disjoint horizontal segments to vertices
and disjoint horizontal or vertical segments to edges so that edges touch their 
endpoints and do not intersect any other vertices.
In \cite{Bie-WAOA12} we showed that
any 2-connected outer-planar graph $G$ has a 
flat visibility representation of height at most $4pw(G)$, 
where the {\em pathwidth} $pw(G)$ is a graph parameter that is a lower bound on
the height of any planar drawing \cite{FLW03}.

By a result of Babu et al.~\cite{BBC+13}, we can add edges to any
outer-planar graph $G$ to obtain a maximal outer-planar graph $G'$ with 
pathwidth in $O(pw(G))$. 
A simple exercise shows that flat visibility
representations can be converted into $y$-monotone poly-line drawings of the 
same height.  
Theorem~\ref{thm:main} hence implies:

\begin{theorem}
Every outer-planar graph $G$ has a planar straight-line drawing of
height $O(pw(G))\subseteq O(\log n)$.
\end{theorem}


\medskip\noindent{\bf Integer programming formulations: }
In a recent paper, we developed integer program (IP) formulations
for many graph drawing problems where vertices and edges are
represented by axis-aligned boxes \cite{BBN+-GD13}.  By adding
some constraints, one can force that edges degenerate to line segments
and vertices to horizontal line segments.  In particular, it is easy to 
create an IP that expresses ``$G$ is drawn as a flat visibility
representation'', using $O(n^3)$ variables and constraints.
%
%
%

It is quite easy to show that every flat visibility representation
can be converted into a $y$-monotone poly-line drawing, and vice versa,
without changing $y$-coordinates.  By Theorem~\ref{th:main} we hence have:

\begin{theorem}
A graph $G$ has a planar straight-line drawing of height $h$
if and only if it has a planar flat visibility representation of height $h$.
\end{theorem}

It is very easy to encode the height in the IP formulations
of \cite{BBN+-GD13}. 
Therefore:

\begin{corollary}
There exists an integer program with $O(hn^2)$ variables and constraints
to test whether a graph $G$ has a planar straight-line drawing of height $h$.
\end{corollary}

While an algorithm was already known to test 
whether $G$ has a planar drawing of height at most $h$ \cite{DFK+08}, 
its rather large run-time of $O(2^{32h^3}\mbox{poly}(n))$
means that solving the above integer program might well be faster
in practice.

\section{Conclusion and open problems}
\label{se:open}

In this paper, we studied how to transform planar poly-line drawings
into straight-line drawings.  In particular we showed 
that $y$-monotone poly-line drawings are no
more powerful (with respect to the height) than straight-line drawings,
because any planar $y$-monotone poly-line drawing can be transformed into
a planar straight-line drawing with the same height.
If we drop ``$y$-monotone'' then we can argue
that poly-line drawings are sometimes truly better (with respect to
height) than straight-line drawings.

We also demonstrated some applications
of our height-preserving transformations, 
especially for obtaining drawings of small height.

Our main open problem concerns the width.
If we want to 
keep the height exactly the same, then exponential width is 
required for the example in Figure~\ref{fig:bad_bad}.  
But is it possible to make the width polynomial
while keeping the height asymptotically the same?
Generally, what is the width of the construction in Theorem~\ref{thm:PL_SL},
and how does it depend on the height?

Also, the graph in Figure~\ref{fig:not_monotone} can easily
be drawn with $y$-monotone curves if we use 7 rows.  
Can every poly-line drawing (not necessarily $y$-monotone)
be converted into a straight-line
drawing with asymptotically the same height?

Finally, are there other invariants that one can maintain
while ``straightening out'' poly-line drawings, at least
under some restrictions such as $y$-monotonicity?



\bibliographystyle{plain}
\bibliography{../../bib/full,../../bib/papers,../../bib/gd}

\begin{thebibliography}{10}

\bibitem{BBC+13}
J.~Babu, M.~Basavaraju, S.~Chandran Leela, and D.~Rajendraprasad.
\newblock 2-connecting outerplanar graphs without blowing up the pathwidth.
\newblock In {\em Computing and Combinatorics (COCOON 2013)}, volume 7936 of
  {\em Lecture Notes in Computer Science}, pages 626--637. Springer, 2013.

\bibitem{Bie-WAOA12}
T.~Biedl.
\newblock A 4-approximation algorithm for the height of drawing 2-connected
  outerplanar graph.
\newblock In {\em Workshop on Approximation and Online Algorithms (WAOA'12)},
  volume 7846 of {\em Lecture Notes in Computer Science}, pages 272--285, 2013.

\bibitem{BKM-WG98}
T.~Biedl, M.~Kaufmann, and P.~Mutzel.
\newblock Drawing planar partitions {II: HH}-drawings.
\newblock In {\em Workshop on Graph-Theoretic Concepts in Computer Science
  (WG'98)}, volume 1517 of {\em Lecture Notes in Computer Science}, pages
  124--136. Springer-Verlag, 1998.

\bibitem{BBN+-GD13}
T.~Biedl, \student{T. Bl{\"a}sius}, \student{B. Niedermann}, M.~N{\"o}llenburg,
  \student{R. Prutkin}, and \student{I. Rutter}.
\newblock Using {ILP/SAT} to determine pathwidth, visibility representations,
  and other grid-based graph drawings.
\newblock In {\em Graph Drawing (GD'13)}, volume 8242 of {\em Lecture Notes in
  Computer Science}, pages 460--471. Springer, 2013.

\bibitem{DBETT98}
G.~Di~Battista, P.~Eades, R.~Tamassia, and I.~Tollis.
\newblock {\em Graph Drawing: Algorithms for Geometric Representations of
  Graphs}.
\newblock Prentice-Hall, 1998.

\bibitem{DFK+08}
V.~Dujmovic, M.~Fellows, M.~Kitching, G.~Liotta, C.~McCartin, N.~Nishimura,
  P.~Ragde, F.~Rosamond, S.~Whitesides, and D.~Wood.
\newblock On the parameterized complexity of layered graph drawing.
\newblock {\em Algorithmica}, 52:267--292, 2008.

\bibitem{Fary48}
I.~F{\'a}ry.
\newblock On straight line representation of planar graphs.
\newblock {\em Acta Scientiarum Mathematicarum (Szeged)}, 11(4):229--233, 1948.

\bibitem{FLW03}
S.~Felsner, G.~Liotta, and S.~Wismath.
\newblock Straight-line drawings on restricted integer grids in two and three
  dimensions.
\newblock {\em Journal of Graph Algorithms and Applications}, 7(4):335--362,
  2003.

\bibitem{FPP90}
H.~de Fraysseix, J.~Pach, and R.~Pollack.
\newblock How to draw a planar graph on a grid.
\newblock {\em Combinatorica}, 10:41--51, 1990.

\bibitem{NR04}
T.~Nishizeki and M.S. Rahman.
\newblock {\em Planar Graph Drawing}, volume~12 of {\em Lecture Notes Series on
  Computing}.
\newblock World Scientific, 2004.

\bibitem{Sch90}
W.~Schnyder.
\newblock Embedding planar graphs on the grid.
\newblock In {\em {ACM}-{SIAM} Symposium on Discrete Algorithms (SODA '90)},
  pages 138--148, 1990.

\bibitem{Stein51}
S.~Stein.
\newblock Convex maps.
\newblock In {\em American Mathematical Society}, volume~2, pages 464--466,
  1951.

\bibitem{Wagner36}
K.~Wagner.
\newblock Bemerkungen zum {V}ierfarbenproblem.
\newblock {\em Jahresbericht der Deutschen Mathematiker-Vereinigung},
  46:26--32, 1936.

\end{thebibliography}

\newpage

\begin{appendix}

\section{The algorithm in \cite{Bie-WAOA12}}
\label{se:VR_SL}

Recall that a {\em flat visibility
representation} is an assignment of disjoint horizontal segments to vertices
and disjoint horizontal or vertical segments to edges so that edges touch their 
endpoints and do not intersect any other vertices.  (In the figures
below, we show vertices thickened into boxes for ease of readability.)
In \cite{Bie-WAOA12}, we claimed the following result. 

\begin{theorem}
\label{thm:VR2SL}
\label{thm:VR_SL}
\label{th:VR2SL}
Any planar flat visibility representation $\Gamma$ can be transformed into a 
planar straight-line drawing $\Gamma'$ with the same $y$-coordinates.
\end{theorem}

While the result is correct (one can easily transform
flat visibility representation into a $y$-monotone poly-line drawing
and then apply Theorem~\ref{thm:PL_SL}), the
simple algorithm that we gave for it in \cite{Bie-WAOA12} unfortunately 
was incorrect.  This section reviews the algorithm and gives the
counter-example.

\subsection{The algorithm} 

Assume that $\Gamma$ is a planar flat visibility representation.
For any vertex $v$, use $x_l(v), x_r(v)$ and $y(v)$ to denote leftmost and
rightmost $x$-coordinate and (unique) $y$-coordinate of the line segment that 
represents $v$ in $\Gamma$.   
Use $X(v)$ and $Y(v)$ to denote
the (to-be-determined) coordinates of $v$ in the straight-line drawing 
$\Gamma'$ that we construct.
For any vertex set $Y(v)=y(v)$, hence $y$-coordinates are the same.

Let $v_1,\dots,v_n$ be the vertices sorted by $x_l(\cdot)$, breaking ties
arbitrarily.  For each vertex $v_i$, let the {\em predecessors}
of $v_i$ be the neighbours of $v_i$ that come earlier in the order
$v_1,\dots,v_n$. 
The algorithm determines $X(\cdot)$ for each vertex by 
processing vertices in  order $v_1,\dots,v_n$
and expanding the drawing $\Gamma'_{i-1}$ created for 
$v_1,\dots,v_{i-1}$ into a drawing $\Gamma'_i$ of $v_1,\dots,v_i$.

Suppose $X(v_g)$ has been computed for all $g<i$ already.  To find
$X(v_i)$, determine lower bounds for it by considering all predecessors
of $v_i$ and taking the maximum over all of them.  
One lower bound for $X(v_i)$ is that it needs to be to the right of
anything in row $y(v_i)$.  Thus, if $\Gamma'_{i-1}$
contains a vertex
or part of an edge at point $(X,y(v_i))$, then $X(v_i)\geq 
\lfloor X \rfloor + 1$ is required.

Next consider any predecessor $v_g$ of $v_i$ with $y(v_g)\neq y(v_i)$.  
Since $v_g$ and $v_i$ are not in the same row, they must see each other
vertically in $\Gamma$, which means that $x_r(v_g)\geq x_l(v_i)$.
See also Figure~\ref{fig:transform}.
So if in $\Gamma$
any vertex $v_k$ exists to the right of $v_g$ in row $y(v_g)$, then
$x_\ell(v_k)> x_r(v_g) \geq x_\ell(v_i)$, which implies that 
$v_k$ has not yet been added to $\Gamma'_{i-1}$.  So vertex
$v_g$ is the rightmost
vertex in its row in $\Gamma'_{i-1}$ and can see towards infinity on the
right.  But then $v_g$ can also see the point $(+\infty,y(v_i))$, or
in other words, there exists some $X_g$ such that $v_g$ can see all
points $(X,y(v_i))$ for $X\geq X_g$.  
Impose the lower bound $X(v_i) \geq \lfloor X_g \rfloor + 1$ on the
$x$-coordinate of $v_i$.

\begin{figure}[t]
\hspace*{\fill}
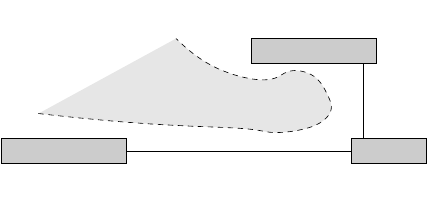
\hspace*{\fill}
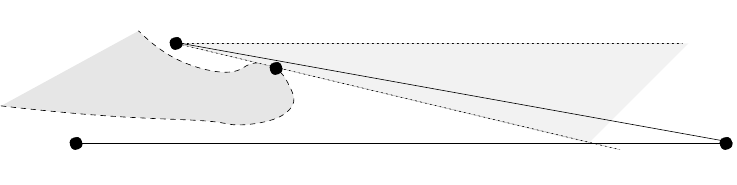
\hspace*{\fill}
\caption{Transforming a flat visibility drawing into a straight-line
drawing with the same $y$-coordinates.}
\label{fig:transform}
\end{figure}

Now let $X(v_i)$ be the smallest value that satisfies the above lower
bounds (from the row $y(v_i)$ and from all predecessors of $v_i$ in
different rows.)  Set $X(v_i)=1$ if there were no such lower bounds.

\subsection{A bad example}

Unfortunately, the above algorithm sometimes does not give planar drawings.
The above algorithm preserves (as one can easily show by induction) the
left-to-right-orders among vertices, but it does not necessarily
preserve left-to-right-orders when we also include the points 
where edges cross rows.

Figure~\ref{fig:algo_wrong} shows a specific example where the
algorithm goes wrong.  The vertices are being placed in order
$\dots,u_1,\dots,x,u_2,\dots,y_1,\dots$.
When placing $u_2$, we also insert the edge
$(u_1,u_2)$, because $u_1$ is a predecessor of $u_2$.  But this edge
is ``much farther right'' in the sense that in the second row from
top, this edge should come {\em after} $y_1$, while the algorithm
draws it before placing $y_1$.  Consequently, the vertices $\{y_1,y_2,w\}$,
which should be drawn inside triangle $\{x,u_1,u_2\}$, are instead
drawn to the right of it, resulting in crossings.

\begin{figure}[t]
\hspace*{\fill}
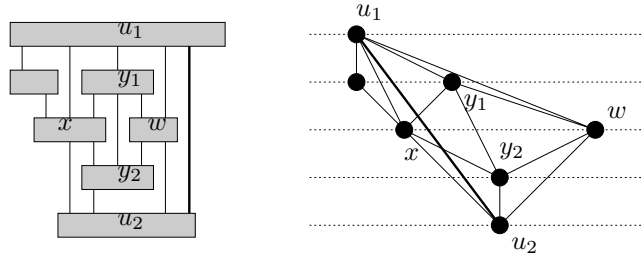
\hspace*{\fill}
\caption{An example where the algorithm from \cite{Bie-WAOA12}
creates a crossing.  Edge $(u_1,u_2)$ (which is drawn ``too early'')
is bold.}
\label{fig:algo_wrong}
\end{figure}

In fact, on this graph the algorithm cannot preserve left-to-right-orders
among edges, even if we used some different processing order 
$v_1,\dots,v_n$. 
For to preserve left-to-right-orders, we would have to
place $(x,u_1)$ {\em before} placing $y_1$,
and so in particular both $x$ and $u_1$ must come before $y_1$ in the
vertex order.  Likewise $x$ and $u_2$ must come before $y_2$ in the vertex 
order.
On the other hand, $w$ must come {\em after} both $y_1$ and $y_2$ due
to edge $(y_1,y_2)$.  Combining the above, we see that both $u_1$ and $u_2$
must come before $w$.  But then edge $(u_1,u_2)$ is drawn before placing
vertex $w$, contradicting the left-to-right-orders  of edges in $w$'s row.

\section{Triangulating a $y$-monotone poly-line drawing}

We now give the full details of how to triangulate a $y$-monotone
poly-line drawing.  We assume that vertices, edges and rows have already
been added so that the outer-face is a triangle. 
Convert---by inserting bends as needed---the given 
drawing into a {\em short}
$y$-monotone poly-line drawing, where any edge-segment is horizontal
or connects two adjacent rows.  
%
Call a poly-line drawing {\em strictly
$y$-monotone} if any edge is either horizontal or drawn with a 
strictly $y$-monotone path.  Any short $y$-monotone poly-line drawing
can be made strict:
If some bend $b$ is incident to a horizontal segment $s_h$
and a non-horizontal segment $s_v$, then connect directly
from the other end $x$ of $s_h$ to the other end $z$ of $s_v$.  Since
$x$ is one row above or below $z$, 
this cannot introduce crossings. 
See Figure~\ref{fig:shortcut}.  

\begin{figure}[t]
\hspace*{\fill}
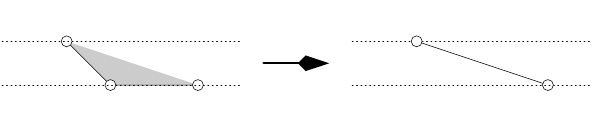
\hspace*{\fill}
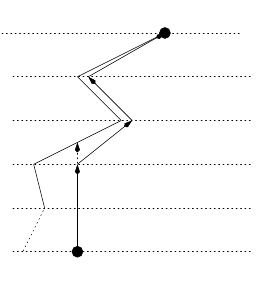
\hspace*{\fill}
\caption{Shortcutting a bend incident to a horizontal edge segment,
and inserting a new path to a vertex higher up.}
\label{fig:shortcut}
\label{fig:new_path}
\end{figure}

Now we have a short strictly $y$-monotone poly-line drawing.
If some inner  vertex $v$ has no neighbour $w$
with $y(w)>y(v)$,  then go upward from $v$ until
we hit some edge $e$, say at $y$-coordinate $Y$.
Let $w$ be the end of $e$ with larger $y$-coordinate (breaking
ties arbitrarily) and add edge $(v,w)$ to the graph.    Add a
strictly $y$-monotone path for $(v,w)$ to $\Gamma$ as follows.
Define $r:=\lceil Y\rceil -1$ to be the row just below where we
hit $e$.
Go upward from $v$ until row $r$, then in rows $r+1,\dots,y(w)-1$
add a bend next to the bend of $e$ (on the side from which we hit $e$),
and then connect to $w$.  See Figure~\ref{fig:new_path}.  No
edge can cross this path that wouldn't have crossed $e$, by
choice of $e$ and the placement of bends.

In this fashion add edges until any inner  vertex 
has neighbours strictly above and below.    
In the resulting drawing 
any inner  face $f$ is a $y$-monotone 
polygon and has no horizontal edges except perhaps a single
horizontal edge each at the
minimum or maximum $y$-coordinate. 
If $f$ contains four or more vertices, then it must contain
two vertices $v,w$ such that $(v,w)$ is not an edge; otherwise we had an
outer-planar drawing of $K_4$.    We now show how to add $(v,w)$ to $\Gamma$.

If $y(v)\neq y(w)$, say $y(v)<y(w)$, then 
find points in rows $y(v)+1,y(v)+2,\dots,y(w)-1$
that are strictly inside $f$.
Since $f$ is $y$-monotone, any two consecutive such points 
can be connected without
crossing, so draw $(v,w)$ following these points.
If $y(v)=y(w)$ then the horizontal
segment from $v$ to $w$ belongs to the closure of the $y$-monotone face $f$. 
In fact it lies strictly inside
$f$ since the only horizontal edge of $f$ are single edges at the
minimum and maximum $y$-coordinate, but $(v,w)$ was not an edge before.
So we can drawn $(v,w)$ horizontally. 

\section{$y$-monotonicity}

\begin{figure}[t]
\hspace*{\fill}
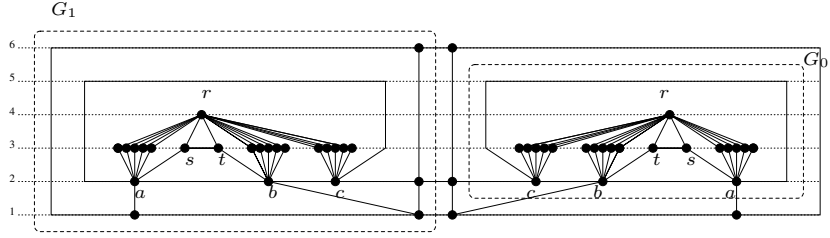
\hspace*{\fill}
\caption{A graph that can be drawn on 6 rows, but not if edges must
be $y$-monotone.}
\label{fig:not_monotone}
\end{figure}

\setcounter{theorem}{3}
\addtocounter{theorem}{-1}

This section gives the full proof that $y$-monotonicity is restrictive,
i.e., that sometimes the height can be reduced if edges are not
drawn $y$-monotone.

\begin{theorem}
There exists a graph with a planar poly-line drawing on 6 rows that
has no planar $y$-monotone poly-line drawing on 6 rows.
\end{theorem}
\begin{proof}
The graph $G$ and the poly-line drawing on 6 rows is shown in 
Figure~\ref{fig:not_monotone}.
Graph $G$ consists of two copies of $G_1$, which in turn consists of a 
4-cycle that surrounds a graph $G_0$.
Graph $G_0$ is constructed from a triangular prism with triangles 
$\{a,b,c\}$ and $\{r,s,t\}$. The edge $(c,r)$
of the triangular prism has been replaced by $K_{2,5}$, and two additional 
copies of $K_{2,5}$ have been inserted at $\{b,r\}$ and $\{a,r\}$.  
Observe that $G$ has only one planar embedding up to symmetry, since its only
cutting pairs are the 2-sides of each $K_{2,5}$.  
We need an observation.

\smallskip\noindent{\bf Claim: }
{\it Any planar drawing of $K_{2,5}$ on three rows contains the
2-side vertices on the top and bottom row.}
For the outer-face of $K_{2,5}$ consists of a 4-cycle, leaving three
vertices (all from the 5-side) that are surrounded by a 4-cycle and
hence on the middle row.
By planarity this forces the vertices of the 2-side onto the bottom row 
and the top row.

\smallskip
Now assume that $\Gamma$ is a $y$-monotone
poly-line drawing of $G$ on 6 rows.  
At least one copy of $G_1$ has the 4-cycle as outer-face in
the induced drawing. 
Because this 4-cycle enclosing $G_0$, the induced drawing
of $G_0$ must exist entirely on rows 2,3,4,5.   Because $G_0$ consists of
a triangle $\{a,b,c\}$ enclosing another triangle, the inner triangle 
$\{r,s,t\}$ is on rows 3 and 4 only.

Triangle $\{a,b,c\}$ encloses a triangle on rows 3 and 4, and hence must
contain points on rows 2 and 5.  By $y$-monotonicity, the maximum and
minimum $y$-coordinate of the triangle must be achieved at vertices,
so one of the vertices, say $a$, has $y$-coordinate 2, while another one,
say $c$, has $y$-coordinate 5.

Assume vertex $r$ had $y$-coordinate 3.  Then the $K_{2,5}$ with 2-side
$\{a,r\}$ is entirely drawn on rows 2,3,4 (because all vertices except
$a$ are inside triangle $\{a,b,c\}$), but $r$ and $a$ are on adjacent
rows.  This contradicts the claim.  If $r$ has $y$-coordinate 4 then
similarly we find a contradiction at the $K_{2,5}$ with 2-side $\{c,r\}$.
So no $y$-monotone poly-line drawing of $G$ on 6 rows can exist.
\qed
\end{proof}

\end{appendix}
\end{document}